\documentclass[a4paper,USenglish]{lipics-v2019}

\usepackage{amsmath}
\usepackage{amssymb}
\usepackage{amsthm}
\usepackage{graphicx}
\usepackage[colorinlistoftodos]{todonotes}
\usepackage{algorithm}
\usepackage[noend]{algpseudocode}
\usepackage{enumerate}
\usepackage{xcolor}
\usepackage{framed}
\usepackage{bm}
\usepackage{pifont}
\usepackage{multicol}
\usepackage{lipsum}
\usepackage{tikz}
\usetikzlibrary{calc}
\usepackage{wrapfig}
\usepackage{array}
\usepackage{relsize}
\usepackage{comment}
\usepackage{url}
\usepackage[title]{appendix}
\usepackage{caption}
\usepackage{subcaption}
\usepackage{xstring}
\usepackage{thmtools,thm-restate}

\usepackage{booktabs}       
\newtheorem{assumption}{Assumption}

\algdef{SE}[RECEIVING]{Receiving}{EndReceiving}[1]{\textbf{upon
receiving}\ #1\ \algorithmicdo}{\algorithmicend\ \textbf{}}%
\algtext*{EndReceiving}

\algdef{SE}[UPON]{Upon}{EndUpon}[1]{\textbf{upon
}\ #1\ \algorithmicdo}{\algorithmicend\ \textbf{}}%
\algtext*{EndUpon}

\algdef{SE}[DELIVERY]{Delivery}{EndDelivery}[1]{\textbf{upon
delivery}\ #1\ \algorithmicdo}{\algorithmicend\ \textbf{}}%
\algtext*{EndDelivery}

\newcommand{\lr}[1]{\langle #1 \rangle}

\title{Not a COINcidence: Sub-Quadratic Asynchronous Byzantine Agreement WHP}

\author{Shir Cohen}{Technion, Israel}{shirco@campus.technion.ac.il}{}{}
\author{Idit Keidar}{Technion, Israel}{idish@ee.technion.ac.il}{}{}
\author{Alexander Spiegelman}{VMware Research, Israel}{sasha.spiegelman@gmail.com}{}{}

\authorrunning{S. Cohen, I. Keidar and A. Spiegelman} 

\Copyright{Shir Cohen, Idit Keidar and Alexander Spiegelman}



\nolinenumbers 

\hideLIPIcs  

\begin{document}

\maketitle

\begin{abstract}
King and Saia were the first to break the quadratic word complexity bound for Byzantine Agreement in synchronous systems against an adaptive adversary, and Algorand broke this bound with near-optimal resilience (first in the synchronous model and then with eventual-synchrony). Yet the question of asynchronous sub-quadratic Byzantine Agreement remained open.
To the best of our knowledge, we are the first to answer this question in the affirmative.
A key component of our solution is a shared coin algorithm based on a VRF. A second essential ingredient is VRF-based committee sampling, which we formalize and utilize in the asynchronous model for the first time.
Our algorithms work against a delayed-adaptive adversary, which cannot perform after-the-fact removals but has full control of Byzantine processes and full information about communication in earlier rounds. Using committee sampling and our shared coin, we solve Byzantine Agreement with high probability, with a word complexity of $\widetilde{O}(n)$ and $O(1)$ expected time, breaking the $O(n^2)$ bit barrier for asynchronous Byzantine Agreement.
\end{abstract}

\section{Introduction}

 Byzantine Agreement (\emph{BA})~\cite{lamport2019byzantine} has been studied for four decades by now, but until recently, has been considered at a fairly small scale. In recent years, however, we begin to see practical use-cases of BA in large-scale systems, which motivates a push for reduced communication complexity. In deterministic algorithms, Dolev and Reischuk’s renown lower bound stipulates that $\Omega(n^2)$ communication is needed \cite{DolevBound}, and until fairly recently, almost all randomized solutions have also had (expected) quadratic word complexity. Recent work has broken this barrier~\cite{king2011breaking,Algorand,nakamoto2012bitcoin}, but not in asynchronous settings. We present here the first sub-quadratic asynchronous Byzantine Agreement algorithm.
 Our algorithm is randomized and solves binary BA \emph{with high probability (whp)}, i.e., with probability that tends to $1$ as $n$ goes to infinity.

We consider a system with a static set of $n$ processes, in the so-called ``permissioned'' setting, where the ids of all processes are well-known. Our algorithm tolerates $f$ failures for $n\approx4.5f$ (asymptotically). In addition, we assume a trusted \emph{public key infrastructure} (PKI) that allows us to use \emph{verifiable random functions} (VRFs)~\cite{micali1999VRF}.

We assume a strong adversary that can adaptively take over processes, whereupon it has full access to their private data. It further sees all messages in the system. But we do limit the adversary in two ways. First, we assume that it is computationally bounded so that we may use the PKI. Second, as proven in~\cite{abraham2019communication} for the synchronous model, achieving sub-quadratic complexity is impossible when the adversary can perform after-the-fact removal, meaning that it can delete messages that were sent by correct processes before corrupting these processes. Here, we adapt the no after-the-fact removal assumption to the asynchronous model, and define a \emph{delayed-adaptive adversary} based on causality~\cite{lamport2019time}.


We formalize the concept of VRF-based committee sampling as used in Algorand~\cite{Algorand,chen2018algorand}, and adapt it to the asynchronous model. 
In a nutshell, the idea is to use a VRF seeded with each process's private key in order to sample uniformly at random
$O(\log n)$ processes for a \emph{committee}, and to have different committees execute different parts of the BA protocol. Each committee is used for sending exactly one protocol message and messages are sent only by committee members, thus reducing the communication cost.
Whereas in Algorand's synchronous model a process can be sure it receives messages from all correct committee members by a timeout, in the asynchronous model this is not the case. Rather, processes make progress by waiting for some threshold number of messages. Without committees, this threshold is normally $n-f$ (waiting for more than $n-f$ processes might violate termination). But since committees are randomly sampled,  we do not know the committee's exact size or the number of Byzantine processes in it. Thus, adapting committees to this model is somewhat subtle and requires ensuring certain conditions regarding the intersection of subsets of committees. In this paper we identify sufficient conditions on sampling, which ensure safety and liveness with high probability. 

Randomized BA algorithms can be seen as if processes toss a random coin at some point during the protocol. While some protocols toss a local coin~\cite{ben1983another,bracha1983resilient} and require exponential expected time to reach agreement, others use the abstraction of a \emph{shared coin}, which involves communication among processes and results in the same coin toss with some well defined \emph{success rate}~\cite{rabin1983randomized,canetti1993fast,cachin2005random,Algorand,king2013byzantine}. 
In this work we present an asynchronous shared coin algorithm that uses a VRF and provides a constant success rate with an equal probability for tossing 0 and 1. Unlike previous shared coin implementations, our solution does not require a priori knowledge of the set of participants, which makes it useful in committee-based constructions.
We then adapt our coin to work with committees and use it to devise a sub-quadratic BA algorithm.


In summary, this paper presents the first formalization of randomly sampled committees using cryptography in asynchronous settings. Based on this technique, it presents the first sub-quadratic asynchronous shared coin and BA whp algorithms. Our algorithms have expected $\widetilde{O}(n)$ word complexity and $O(1)$ expected time.

{\bf Roadmap.} The rest of this paper is organized as follows. Section \ref{sec:model} describes the model;  Section \ref{sec:related_work} reviews related work. In Section \ref{sec:shared_coin}, we present our shared coin algorithm and in Section \ref{sec:committees}, we formalize committee sampling. Then, in Section \ref{sec:consensus}, we use the coin and the committee sampling to construct a BA whp algorithm. We end with some concluding remarks
in Section \ref{sec:conclusions}.

\section{Model and Preliminaries}
\label{sec:model}

We consider a distributed system consisting of a well-known static set $\Pi$ of $n$ processes and a \emph{delayed-adaptive adversary} (see definition below). 
The adversary may adaptively corrupt up to  $f=(\frac{1}{3}-\epsilon)n$ processes in the course of a run, where $\max \{\frac{3}{8\ln n},0.109\}+\frac{1}{8\ln n}<\epsilon < \frac{1}{3}$.
A corrupted process is \emph{Byzantine}; it may deviate arbitrarily from the protocol. In particular, it may crash, fail to send or receive messages, and send arbitrary messages. As long as a process is not corrupted by the adversary, it is \emph{correct} and follows the protocol.



{\bf Delayed-adaptive adversary.} In the synchronous model, one defines a \emph{late} adversary~\cite{robinson2018breaking, klonowski2019ordered,awerbuch2007denial,ahmadi2020cost}, which at the beginning of round $r$, can observe the state of the system at the beginning of round $r-1$. This assumption prevents “after-the-fact” removals of messages sent by processes before being taken over by the adversary~\cite{abraham2019communication,Algorand}, as required for achieving a sub-quadratic communication cost. We adapt this assumption to the asynchronous model. Since in asynchronous models the natural order between messages is Lamport's happens-before relation~\cite{lamport2019time}, we use the notion of causality instead of `rounds' to define what messages the adversary may observe when scheduling other messages. We denote by $m\rightarrow m'$ the fact that $m$ causally precedes $m'$. The adversary is formally defined as follows:

\begin{definition}[delayed-adaptive adversary]

The delayed-adaptive adversary may adaptively corrupt up to $f$ processes over the course of a run and schedules all messages. The adversary has full access to corrupted processes' private information and can observe all communication, but it can use the contents of a message $m$ sent by a correct process for scheduling a message $m'$ only if $m\rightarrow{m'}$.
\end{definition}

In addition, we assume that once the adversary takes over a process, it cannot “front run” messages that that process had already sent when it was correct, causing the correct messages to be supplanted. Blum et al.~\cite{cryptoeprint:2020:851} achieve this property by using a separate key to encrypt each message, and deleting the secret key immediately thereafter.



{\bf Cryptographic tools.} We assume a trusted PKI, where private and public keys for the processes are generated before the protocol begins and processes cannot manipulate their public keys. In addition, we assume that the adversary is computationally bounded, meaning that it cannot obtain the private keys of processes unless it corrupts them. Furthermore, we assume that the PKI is in place from the outset. (Recall that we assume a permissioned setting, so the public keys of the $n$ processes are well-known). These assumptions allow us to use verifiable random functions, as we now define.

\newcommand{\VRF}{\mathrm{VRF}}
\newcommand{\VRFVer}{\mathrm{VRF\text{-}Ver}}

A \emph{verifiable random function ($\VRF$)} is a pseudorandom function that provides a proof of its correct computation~\cite{micali1999VRF}.
Given a secret key $sk$, one can evaluate the VRF on any input $x$ and obtain a pseudorandom output $y$ together with a proof $\pi$, i.e., $\lr{y,\pi}=\VRF_{sk}(x)$.
From $\pi$ and the corresponding public key $pk$, one can verify that $y$ is correctly computed from $x$ and $sk$ using the function $\VRFVer_{pk}(x,\lr{y,\pi})$. 
Additionally, a VRF needs to satisfy \emph{uniqueness}. More formally, a VRF guarantees the following properties:

  \begin{itemize}
    
	\item Pseudorandomness: for any $x$, it is infeasible to distinguish $y=\VRF_{sk}(x)$ from a uniformly random value without access to $sk$.
    
	\item Verifiability: $\VRFVer_{pk}(x,\VRF_{sk}(x)) = true$. 
    
    \item Uniqueness: it is infeasible to find $x, y_1, y_2, \pi_1, \pi_2$ such that $y_1 \neq y_2$ but 
		\-$\VRFVer_{pk}(x, \-\lr{y_1,\pi_1})=\VRFVer_{pk}(x, \lr{y_2,\pi_2})=true$.
  \end{itemize}

Efficient constructions for VRFs have been described in the literature~\cite{dodis2005VRF,Franklin2013VRF}.


{\bf Communication.} We assume that every pair of processes is connected via a reliable link. Messages are authenticated in the sense that if a correct process $p_i$ receives a message $m$ indicating that $m $ was sent by a correct process $p_j$, then $m$ was indeed generated by $p_j$ and sent to $p_i$. The network is asynchronous, i.e., there is no bound on message delays.


{\bf Complexity.} We use the following standard complexity notions~\cite{abraham2019asymptotically, mostefaoui2015signature}.
While measuring complexity, we allow a \emph{word} to contain a signature, a VRF output, or a value from a finite domain. We define the \textit{duration} of an execution as the longest sequence of messages that are causally related in this execution until all correct processes decide.
We measure the expected \emph{word communication complexity} of our protocols as
the maximum of the expected total number of words sent by correct processes and the expected \emph{running time} of our protocol as the maximum of the expected duration.
In both cases the maximum is computed over all inputs and applicable adversaries and expectation is taken over the random VRF outputs.
 
\section{Related Work}
\label{sec:related_work}

{\bf Lower bounds.} 
Our assumptions conform with a number of known bounds. Deterministic consensus is impossible in an asynchronous system if even one process may crash (by FLP~\cite{fischer1985impossibility}) and requires $\Omega(n^2)$ communication even in synchronous systems~\cite{DolevBound}. As for randomized Byzantine Agreement, Abraham et al.
\cite{abraham2019communication} state that disallowing
after-the-fact removal is necessary even in synchronous settings for achieving sub-quadratic communication.

{\bf Asynchronous BA and shared coin algorithms.} 
The algorithms we present in this paper belong to the family of asynchronous BA algorithms, which sacrifice determinism in order to circumvent FLP. We compare our solutions to existing ones in Table~\ref{tab:survey}.
 
Ben-Or~\cite{ben1983another} suggested a protocol with resilience $n>5f$. This protocol uses a local coin (namely, a local source of randomness) and its expected time complexity is exponential (or constant if $f=O(\sqrt n)$). Bracha~\cite{bracha1987asynchronous} improved the resilience to $n>3f$ with the same complexity. The complexity can be greatly reduced by replacing the local coin with a shared one with a guaranteed success rate.

Later works presented the shared coin abstraction and used it to solve BA with $O(n^2)$ communication. Rabin~\cite{rabin1983randomized} was the first to do so, suggesting a protocol with resilience $n>10f$ and a constant expected number of rounds. Cachin et al.~\cite{cachin2005random} were the first to use a shared coin to solve BA with $O(n^2)$ communication and optimal resilience.
Mostefaoui et al.~\cite{mostefaoui2015signature} then presented a signature-free BA algorithm with optimal resilience and $O(n^2)$ messages that uses a shared coin abstraction as a black box; the shared coin algorithm we provide in Section 4 can be used to instantiate this protocol. All of the aforementioned algorithms solve binary BA, where the processes’ initial values are 0 and 1; a recent work solved multi-valued BA with the same $O(n^2)$ word complexity~\cite{abraham2019asymptotically}.

BA algorithms also differ in the cryptographic assumptions they make and the cryptographic tools they use. Rabin’s coin~\cite{rabin1983randomized} is based on cryptographic secret sharing~\cite{shamir1979share}. Some later works followed suit, and used cryptographic abstractions such as threshold signatures~\cite{abraham2019asymptotically,cachin2005random}. Other works forgo cryptography altogether and instead consider a full information model, where there are no restrictions on the adversary’s computational power~\cite{canetti1993fast,king2013byzantine}. In this model, the problem is harder, and existing works achieve very low resilience~\cite{king2013byzantine} ($n>400f$) or high communication complexity~\cite{canetti1993fast}. In this paper we do use cryptographic primitives. We assume a computationally bounded adversary and rely on the abstraction of a VRF~\cite{micali1999VRF}. VRFs were previously used in blockchain protocols~\cite{Algorand,hanke2018dfinity,Helix} and were also used by Micali~\cite{DBLP:conf/innovations/Micali17} to construct a shared coin in the synchronous model.

Several works~\cite{Abraham2018HotStuffTL,baudet2019state,kwon2014tendermint,naor2019cogsworth,spiegelman2020search} solve BA with subquadratic complexity in the so-called optimistic case (or ``happy path''), when communication is timely and a correct process is chosen as a ``leader''. In contrast, we focus on the worst-case asynchronous case.

\begin{table}
\caption{Asynchronous Byzantine Agreement algorithms.}

\setlength\extrarowheight{3pt}

\begin{tabular}{lccccc}

\toprule
Protocol & n $>$ & Adversary & Word complexity & Termination & Safety \\

\midrule

Ben-Or~\cite{ben1983another}   & $5f$ & adaptive & $O(2^n)$ & w.p. 1 & $\checkmark$ \\

Rabin~\cite{rabin1983randomized}   & $10f$ & adaptive & $O(n^2)$ & w.p. 1 & $\checkmark$ \\ 

Bracha~\cite{bracha1987asynchronous} & $3f$  & adaptive & $O(2^n)$ & w.p. 1 & $\checkmark$ \\ 

Cachin et al.~\cite{cachin2005random} & $3f$  & adaptive & $O(n^2)$ &  w.p. 1 & $\checkmark$ \\ 

King-Saia~\cite{king2013byzantine} & $400f$  & adaptive & polynomial & whp & $\checkmark$ \\ 

MMR~\cite{mostefaoui2015signature} & $3f$  & adaptive & $O(n^2)$ & w.p. 1 & $\checkmark$ \\ 


Our protocol & $\approx4.5f$  & delayed-adaptive & $\tilde{O}(n)$ & whp & whp \\ 

\bottomrule

\end{tabular}
\label{tab:survey}
\end{table}

{\bf Committees.} 
We use committees in order to reduce the word complexity and allow each step of the protocol to be executed by only a fraction of the processes. King and Saia used a similar concept and presented the first sub-quadratic BA protocol in the synchronous model~\cite{king2011breaking}. Algorand proposed a synchronous algorithm~\cite{Algorand} (and later extended it to eventual synchrony~\cite{chen2018algorand}) where committees are sampled randomly
using a VRF. Each process executes a local computation to sample itself to a committee, and hence the selection of processes does not require interaction among them. We follow this approach in this paper and adapt the technique to the asynchronous model.

Following initial publication of our work, Blum et al.~\cite{cryptoeprint:2020:851} have also achieved subquadratic BA WHP under an adaptive adversary. Their assumptions are incomparable to ours -- while they strengthen the adversary to remove the delayed adaptivity requirement, they also strengthen the trusted setup. Specifically, they use a trusted dealer to a priori determine the committee members, flip the shared coin, and share it among the committee members. In contrast, we use a peer-to-peer protocol to generate randomness, and require delayed adaptivity in order to prevent the adversary from tampering with this randomness. As in our protocol, setup has to occur once and may be used for any number of BA instances.

\section{Shared Coin}
\label{sec:shared_coin}

We describe here an asynchronous protocol for a shared coin with a constant success rate against the delayed-adaptive adversary. We assume that for every $r\in \mathbb{N}$, shared\textunderscore coin$(r)$ is invoked by all correct processes and that the invocation of shared\textunderscore coin$(r)$ by some process $p$ is causally independent of its progress at other processes.
The definition of a shared coin is given below.



\begin{definition}[Shared Coin]

A shared coin with success rate $\rho$ is a shared object that generates an infinite sequence of binary outputs. For each execution of the procedure shared\textunderscore coin$(r)$ with 
$r\in \mathbb{N}$, all correct processes output $b$ with probability at least $\rho$, for any value $b\in\{0,1\}$.
\end{definition}

The pseudo-code for our shared coin is presented in Algorithm \ref{alg:shared_coin_protocol}. Our protocol is composed of two phases of messages passing. Each process first samples the VRF with its private key and the protocol's argument in order to generate a random initial value. For brevity, we denote by $VRF_i$ the VRF with $p_i$'s private key. Using a VRF to generate a random initial value effectively weakens the adversary as Byzantine processes can neither choose their initial values nor equivocate. If a Byzantine process would try to act maliciously, the VRF proof would easily expose it and its message would be ignored.

In each phase of the protocol, each process sends one value to every other process. The receiver validates the received values using the VRF proofs, which are sent along with the values. We omit the proof validation from the code for clarity. After two phases of communication, each process chooses the minimum value it received in the second phase and outputs its least significant bit.
We follow the concept of a common core, as presented by Attiya and Welch for the crash failure model \cite{attiya2004distributed}, and argue that if a core of $f+1$ correct processes hold the global minimum value at the end of phase 1, then by the end of the following phase all processes receive this value.
We exploit the $\epsilon$ parameter in our resilience definition to bound the number of values held by $f+1$ correct processes. We show that this number is linear in $n$ and hence with a constant positive probability, by the end of the second phase, all correct processes receive the global minimum among the VRF outputs and therefore produce the same output.

\begin{algorithm}
\caption{Protocol shared\textunderscore coin($r$): code for process $p_i$}

\begin{algorithmic}[1]

\State Initially \emph{first-set, second-set} $=\emptyset$

\State
$v_i \gets \emph{VRF}_i(r)$
\State send $\lr{\textsc{first},v_i}$
to all processes 
\Statex

\Receiving{ $\lr{\textsc{first},v_j}$ with valid $v_j$
from $p_j$ }

    	\If{$v_j < v_i$} $v_i \gets v_j$
	    \EndIf
    	
    	\State $\emph{first-set} \gets \emph{first-set} \cup \{j\}$   
    	\State \textbf{when} |\emph{first-set}| $=n-f$ for the first time \label{l.wait_first}
    	\State \hspace{\algorithmicindent} send $\lr{\textsc{second},v_i}$  to all processes  \label{l.send_second}

\EndReceiving
\Statex

\Receiving{ $\lr{\textsc{second},v_j}$ with valid $v_j$ from $p_j$}

   	\If{$v_j < v_i$} $v_i \gets v_j$
    \EndIf
	\State $\emph{second-set} \gets \emph{second-set} \cup \{j\}$   
	
	\State \textbf{when} $|\emph{second-set}|$ $=n-f$ for the first time
    \State \hspace{\algorithmicindent} \textbf{return} \emph{LSB}$(v_i)$

\EndReceiving

\end{algorithmic}
\label{alg:shared_coin_protocol}
\end{algorithm}

We now prove that the shared coin has a constant success rate.
We say that a value $v$ is \emph{common} if at least $f+1$ correct processes receive $v$ by the end of phase 1. Denote by $c$ be the number of different common values.
The next two lemmas give a lower bound on $c$ and on the probability that the global minimum among the VRF outputs is common.

\begin{lemma}
\label{common_vals}
In Algorithm \ref{alg:shared_coin_protocol}, $c \geq \frac{9\epsilon}{1+6\epsilon}n$.
\end{lemma}

\begin{proof}
In a given run of the algorithm, define a table T with $n$ rows and $n$ columns, where for each correct process $p_i$ and each $0\leq j\leq n-1$, $T[i,j]=1$ iff $p_i$ receives $\lr{\textsc{first},v}$ from $p_j$ before sending the second message in line \ref{l.send_second}. Each row of a correct process contains exactly $n-f$ ones since it waits for $n-f$ $\lr{\textsc{first},v}$ messages (line \ref{l.wait_first}). Each row of a faulty process is arbitrarily filled with $n-f$ ones and $f$ zeros. Thus, the total number of ones in the table is $n(n-f)$ and the total number of zeros is $nf$.
Let $k$ be the number of columns with at least $2f+1$ ones. Because each column represents a value and out of the $2f+1$ ones at least $f+1$ represent correct processes that receive this value, $c\geq k$.
Denote by $x$ the number of ones in the remaining columns. Because each column has at most $n$ ones we get:
\begin{equation}
x \geq n(n-f)-kn.
\end{equation}
And because the remaining columns have at most $2f$ ones:
\begin{equation}
x \leq 2f(n-k).
\end{equation}

Combining $(1),(2)$ we get:
\begin{equation*}
2f(n-k) \geq n(n-f)-kn
\end{equation*}
\begin{equation*}
2fn-2fk \geq n^2-fn-kn
\end{equation*}
\begin{equation*}
(n-2f)k \geq n^2-3fn
\end{equation*}
\begin{equation*}
k \geq \frac{n(n-3f)}{n-2f}.
\end{equation*}

Because $f=(\frac{1}{3}-\epsilon)n$ we get:
\begin{equation*}
\label{eq:4}
c\geq k\geq \frac{n(n-3 (\frac{1}{3}-\epsilon)n)}{n-2 (\frac{1}{3}-\epsilon)n} = 
\frac{n(1-1+3\epsilon)}{1-\frac{2}{3}+2\epsilon} =
\frac{9\epsilon}{1+6\epsilon}n, \textrm{as required.}
\end{equation*} 
\end{proof}



Let $v_{min}\triangleq \displaystyle \min_{ p_i \in \Pi}\{VRF_i(r)\}$. We prove that with a constant probability, it is common.

\begin{lemma}
\label{common_prob}
 $Prob[v_{min}\;is\;common]\geq \frac{c}{n}-\frac{1}{3}+\epsilon$.
\end{lemma}

\begin{proof}
Notice that we assume that the invocation of shared\textunderscore coin$(r)$ by each process is causally independent of its progress at other processes. Hence, for any two processes $p_i, p_j$, the messages $\lr{\textsc{first},v_i}$, $\lr{\textsc{first},v_j}$ are causally concurrent. Thus, due to our \emph{delayed-adaptive adversary} definition,
these messages are scheduled by the adversary regardless of their content, namely their VRF random values.
Notice that the adversary can corrupt processes before they initially send their VRF values. Since the adversary cannot predict the VRF outputs of the processes, the probability that the process holding $v_{min}$ is corrupted before sending its \textsc{first} messages is at most $\frac{f}{n}$.
The adversary is oblivious to the correct processes' VRF values when it schedules their first phase messages. Therefore, each of them has the same probability to become common. Since at most $f$ common values originate at Byzantine processes, this probability is at least $\frac{c-f}{n-f}$.
 We conclude that $v_{min}$ is common with probability at least $(1-\frac{f}{n})\frac{c-f}{n-f}=(1-\frac{(\frac{1}{3}-\epsilon)n}{n})\frac{c-(\frac{1}{3}-\epsilon)n}{n-(\frac{1}{3}-\epsilon)n}=(\frac{2}{3}+\epsilon)\frac{c-(\frac{1}{3}-\epsilon)n}{(\frac{2}{3}+\epsilon)n}=\frac{c-(\frac{1}{3}-\epsilon)n}{n}=\frac{c}{n}-\frac{1}{3}+\epsilon$.
 
\end{proof}

We next observe that if $v_{min}$ is common, then it is shared by all processes.

\begin{lemma}
\label{common_global_min}
If $v_{min}$ is common then each correct process holds $v_{min}$ at the end of phase 2.
\end{lemma}
\begin{proof}
Since $v_{min}$ is common, at least $f+1$ correct processes receive it by the end of phase 1 and update their local values to $v_{min}$. During the second phase, each correct process hears from $n-f$ processes. This means that it hears from at least one correct process that has updated its value to $v_{min}$ and sent it.
\end{proof}

\begin{lemma}
\label{safety}
The coin's success rate is at least $\frac{18\epsilon^2+24\epsilon-1}{6(1+6\epsilon)}$.
\end{lemma}

\begin{proof}
We bound the probability that all correct processes output $b\in \{0,1\}$ as follows:

 $Prob[$all\;correct\;processes\;output\;$b$ $]\geq Prob[$all\;correct\;processes\;have\;the\;same\;$v_i$ at the end of phase 2\;and\;its\;LSB\;is\;$b]\geq Prob[$all\;correct\;processes have $v_i=v_{min}$ at the end of phase 2 and its LSB is $b]=\frac{1}{2} \cdot Prob[$all\;correct\;processes have $v_i=v_{min}$ at the end of phase 2] $\stackrel{\text{Lemma \ref{common_global_min}}}{\geq}
 \frac{1}{2} \cdot Prob[v_{min}$ is common$]\stackrel{\text{Lemma \ref{common_prob}}}{\geq}
 \frac{1}{2} (\frac{c}{n}-\frac{1}{3}+\epsilon)
 \stackrel{\text{Lemma \ref{common_vals}}}{\geq}    \frac{18\epsilon^2+24\epsilon-1}{6(1+6\epsilon)}$.

\end{proof}

\begin{remark}
Notice that for $\epsilon=\frac{1}{3}$ (i.e., $f=0$)  it holds that the coin's success rate is $\frac{1}{2}$ and we get a perfect fair coin.
\end{remark}

We have shown a bound on the coin's success rate in terms of $\epsilon$. Since $\epsilon>0.109$, the coin's success rate is a positive constant.
We next prove that the coin ensures liveness.
 
\begin{lemma}
\label{shared_coin_termination}
If all correct processes invoke Algorithm \ref{alg:shared_coin_protocol} then all correct processes return.
\end{lemma}

\begin{proof}
All correct processes send their messages in the first phase. As up to $f$ processes may be faulty, each correct process eventually receives $n-f$ $\lr{\textsc{first},x}$ messages and sends a message in the second phase. As $n-f$ correct processes send their messages, each correct process eventually receives $n-f$ $\lr{\textsc{second},x}$ messages and returns.
\end{proof}

From Lemma \ref{safety} and Lemma \ref{shared_coin_termination} we conclude:
\begin{theorem}
 Algorithm \ref{alg:shared_coin_protocol} implements a shared coin with success rate at least $\frac{18\epsilon^2+24\epsilon-1}{6(1+6\epsilon)}$.
\end{theorem}

{\bf Complexity.} 
In each shared coin instance all correct processes send two messages to all other processes. Each of these messages contains one VRF output (including a value and a proof), in addition to a constant number of bits that identify the message's type. Therefore, each message's size is a constant number of words and the total word complexity of a shared coin instance is $O(n^2)$.

We have presented a new shared coin in the asynchronous model that uses a VRF. This coin can be incorporated into the Byzantine Agreement algorithm of Mostefaoui et al.~\cite{mostefaoui2015signature}, to yield an asynchronous binary Byzantine Agreement with resilience $f=(\frac{1}{3}-\epsilon)n$, a word complexity of $O(n^2)$, and $O(1)$ expected time.


\section{Committees}
\label{sec:committees}

\subsection{Validated committee sampling}

With the aim of reducing the number of messages and achieving sub-quadratic word complexity, it is common to avoid all-to-all communication phases \cite{Algorand,king2011breaking}. Instead, a  subset of processes is sampled to a committee and only processes elected to the committee send messages. As committees are randomly sampled, preventing the adversary from corrupting their members, each committee member cannot predict the next committee sample and send its message to all other processes. Potentially, if the committee is sufficiently small, this technique allow committee-based protocols to result in sub-quadratic word complexity.

Using VRFs, it is possible to implement \emph{validated committee sampling}, which is a primitive that allows processes to elect committees without communication and later prove
their election. It provides every process $p_i$ with a private
function $\emph{sample}_i(s,\lambda)$, which gets a string $s$ and a
threshold $1\leq \lambda \leq n$ and returns a tuple $\lr{v_i,\sigma_i}$, where
$v_i \in \{ \emph{true},\emph{false}\}$ and $\sigma_i$ is a proof that $v_i=\emph{sample}_i(s,\lambda)$. If $v_i=\emph{true}$ we say that $p_i$ is \emph{sampled} to the committee for $s$ and $\lambda$. The primitive ensures that $p_i$ is sampled with probability $\frac{\lambda}{n}$.
In addition, there is a public (known to all) function,
$\emph{committee-val}(s,\lambda,i,\sigma_i)$, which gets a
string $s$, a threshold $\lambda$, a process identification $i$ and a proof
$\sigma_i$, and returns \emph{true} or \emph{false}.

Consider a string $s$.
For every $i$, $1 \leq i \leq n$, let $\lr{v_i,\sigma_i}$ be the return value of $\emph{sample}_i(s,\lambda)$. The following is satisfied for every $p_i$:

\begin{itemize}
  
  \item $\emph{committee-val}(s,\lambda,i,\sigma_i) = v_i$.
  
  \item If $p_i$ is correct, then it is infeasible
  for the adversary to compute $\emph{sample}_i(s,\lambda)$.
  
  \item It is infeasible for the adversary to find  $\lr{v, \sigma}$
  s.t.\ $v \neq v_i$ and
  $\emph{committee-val}(s,\lambda,i,\sigma) = true$.
  
\end{itemize}


We refer to the set of processes sampled to the committee for $s$ and $\lambda$ as $C(s,\lambda)$. In this paper we set $\lambda$ to $8\ln n$.
Let $d$ be a parameter of the system such that $\max \{\frac{1}{\lambda},0.0362\}<d < \frac{\epsilon}{3}-\frac{1}{3\lambda}$.
Our committee-based protocols can no longer wait for $n-f$ processes. Instead, they wait for $W\triangleq \left \lceil{(\frac{2}{3}+3d)\lambda}\right \rceil$ processes. We show that whp at least $W$ processes will be correct in each committee sample and hence waiting for this number does not compromise liveness. In addition, instead of assuming $f$ Byzantine processes, our committee-based protocols assume that whp the number of Byzantine processes in each committee is at most $B\triangleq \left \lfloor{(\frac{1}{3}-d)\lambda}\right \rfloor$.
The following claim is proven in Appendix \ref{chernoof_appendix} using Chernoff bounds.

\begin{restatable}{claim}{resampling}
\label{sampling}
For a string $s$ and $\lambda=const\cdot \ln n$ the following hold with high probability:

\begin{description}
  \item[(S1)] $|C(s,\lambda)|\leq (1+d)\lambda$.
  \item[(S2)] $|C(s,\lambda)|\geq (1-d)\lambda$.
  \item[(S3)] At least
  $W$ processes in $C(s,\lambda)$ are correct.
    \item[(S4)] At most $B$ processes in $C(s,\lambda)$ are Byzantine.
\end{description}

\end{restatable}
If a protocol uses a constant number of committees, then with high probability, Claim \ref{sampling} holds for all of them.
If, however, a protocol uses a polynomial number of committees then it does not guarantee the properties of this claim.
The following corollaries are derived from Claim \ref{sampling} and are used to ensure the safety and liveness properties of our protocols that use committees (a full proof is in Appendix \ref{chernoof_appendix}). Intuitively, S3 allows the protocol to wait for $W$ messages without forgoing liveness. Property S5 below shows that if two processes wait for sets $P_1$ and $P_2$ of this size, then they hear from at least $B+1$ common processes of which, by S4, at least one is correct.

\begin{restatable}[S5]{corollary}{WW}
\label{S5}
Consider $C(s,\lambda)$ for some string $s$ and some $\lambda=const\cdot \ln n$ and two sets $P_1,P_2\subset C(s,\lambda)$ s.t $|P_1|=|P_2|=W$. Then, $|P_1\cap P_2|\geq B +1$.
 
\end{restatable}

The following property is used to show that if $B+1$ correct processes hold some value, and some correct process waits for messages from $W$ processes, then it hears from at least one correct process that holds this value.

\begin{restatable}[S6]{corollary}{WB}
\label{S6}
Consider $C(s,\lambda)$ for some string $s$ and some $\lambda=const\cdot \ln n$ and two sets $P_1,P_2\subset C(s,\lambda)$ s.t $|P_1|=B+1$ and $|P_2|=W$. Then, $|P_1\cap P_2|\geq 1$.
\end{restatable}



\subsection{WHP Coin}
\label{sec:shared_coin_comm}

We now employ committee sampling to reduce the word complexity of our shared coin. Our new protocol is called whp\textunderscore coin. As before, we assume that for every $r\in \mathbb{N}$, the invocation of whp\textunderscore coin$(r)$ by some process $p$ is causally independent of its progress at other processes. We now define the \emph{WHP coin} abstraction:

\begin{definition}[WHP Coin]
\label{def:whp_coin}
A WHP coin with success rate $\rho$ is a shared object exposing whp\textunderscore coin$(r)$, $r\in \mathbb{N}$ at each process.
If all correct processes invoke whp\textunderscore coin$(r)$ then, whp (1) all correct processes return, and (2) all of them output the same value $b$ with probability at least $\rho$, for any value $b\in\{0,1\}$.
\end{definition}

The whp\textunderscore coin protocol is presented in Algorithm \ref{alg:shared_coin_comm_protocol}. It samples two committees, one for each communication step. In each step, only the processes that are sampled to the committee send messages. However, since the committee samples are unpredictable, messages are sent to all processes. With committees, processes can no longer wait for $n-f$ messages.
Instead they wait for $W$ messages.
Since a constant number of committees is sampled in the protocol, Claim \ref{sampling} holds for all of them and by S3, all processes receive $W$ messages, ensuring liveness.

\begin{algorithm}
\caption{Protocol whp\textunderscore coin($r$): code for process $p_i$}

\begin{algorithmic}[1]

\State Initially \emph{first-set, second-set} $=\emptyset$, $v_i=\infty$

\If{$\emph{sample}_i(\textsc{first},\lambda) = $ \emph{true}} \label{l.com.sample1}
    \State
    $v_i \gets \emph{VRF}_i(r)$
    \State send $\lr{\textsc{first},v_i}$
    to all processes 
\EndIf
\Statex

\Receiving{ $\lr{\textsc{first},v_j}$ with valid $v_j$

from validly sampled $p_j$ }
    \If{$\emph{sample}_i(\textsc{second},\lambda)$} \label{l.com.sample2}

            
         	\If{$v_j < v_i$} $v_i \gets v_j$
	        \EndIf

        	\State \emph{first-set} $\gets$ \emph{first-set} $\cup \{j\}$  
        	\State \textbf{when} |\emph{first-set}| $=W$ for the first time \label{l.com.wait_first}
	    	\State \hspace{\algorithmicindent} send $\lr{\textsc{second},v_i)}$  to all processes \label{l.com.send_second}

    \EndIf
\EndReceiving

\Statex

\Receiving{ $\lr{\textsc{second},v_j}$ with valid $v_j$

from validly sampled $p_j$}
 	\If{$v_j < v_i$} $v_i \gets v_j$
    \EndIf

	\State \emph{second-set} $\gets$ \emph{second-set} $\cup \{j\}$

	\State \textbf{when} $|\emph{second-set}|=W$ for the first time
    \State \hspace{\algorithmicindent} \textbf{return} \emph{LSB}$(v_i)$
\EndReceiving

\end{algorithmic}

\label{alg:shared_coin_comm_protocol}
\end{algorithm}

In Appendix \ref{whp_coin_appendix} we adapt the coin’s correctness proof given in Section \ref{sec:shared_coin} to the committee-based protocol, proving the following theorem:

\begin{restatable}{theorem}{rewhpcoin}
 Algorithm \ref{alg:shared_coin_comm_protocol} implements a WHP coin with a constant success rate.
\end{restatable}

{\bf Complexity.} 
In each whp\textunderscore coin instance using committees all correct processes that are sampled to the two committees (lines \ref{l.com.sample1},\ref{l.com.sample2}) send messages to all other processes. Each of these messages contains a VRF output (including a value and a proof), a VRF proof of the sender's election to the committe and a constant number of bits that identify the type of message that is sent. Therefore, each message's size is a constant number of words and the total word complexity of a WHP coin instance is $O(nC)$ where $C$ is the number of processes that are sampled to the committees. Since each process is sampled to a committee with probability $\frac{1}{\lambda}$, we get a word complexity of $O(n\lambda)=O(n\log n)=\widetilde{O}(n)$ in expectation.

\section{Asynchronous sub-quadratic Byzantine Agreement}
\label{sec:consensus}

We adapt the Byzantine Agreement algorithm of Mostefaoui et al.~\cite{mostefaoui2015signature} to work with committees. Our protocol leverages 
an \emph{approver} abstraction, which we implement in Section \ref{sec:approver} and then integrate it into a Byzantine Agreement protocol in Section \ref{sec:consensus_protocol}. 


\subsection{Approver abstraction}
\label{sec:approver}

The \emph{approver} abstraction is an adaptation of the Synchronized Binary-Value Broadcast (SBV-broadcast) primitive in \cite{mostefaoui2015signature}. It provides processes with the procedure \emph{approve(v)}, which takes a value $v$ as an input and returns a set of values.

\begin{assumption}
\label{two_values}
Correct processes invoke the approver with at most $2$ different values.
\end{assumption}

Under this assumption, an approver satisfies the following:

\begin{definition}[Approver]
\label{def:approver}
In an approver instance the following properties hold whp:
\begin{itemize}
  \item Validity. If all correct processes invoke approve(v) then the only possible return value of correct processes is $\{v\}$.
  
  \item Graded Agreement. If a correct process $p_i$ returns $\{v\}$ and another correct process $p_j$ returns $\{w\}$ then $v=w$.
  
  \item Termination. If all correct processes invoke approve then approve returns with a non-empty set at all of them.
\end{itemize}
\end{definition}

Our approver uses different committees for different message types, as illustrated in Fig. \ref{fig:comm_in_app}. Importantly, the protocol satisfies the so-called \emph{process replaceability}~\cite{Algorand} property, whereby a correct process selected for a committee $C$ broadcasts at most one message in its role as a member of $C$. Thus, our delayed-adaptive adversary can learn of a process’s membership in a committee only after that process ceases to partake in the committee. This allows us to leverage the sampling analysis in the previous section. For clarity of the presentation, we discuss the algorithm here under the assumption that properties S1-S6 hold for all sampled committees. As shown above, these hold whp for each committee, and the algorithm employs a constant number of committees, so they hold for all of them whp.

 \begin{figure}[htp]
\centering
\includegraphics[width=7cm]{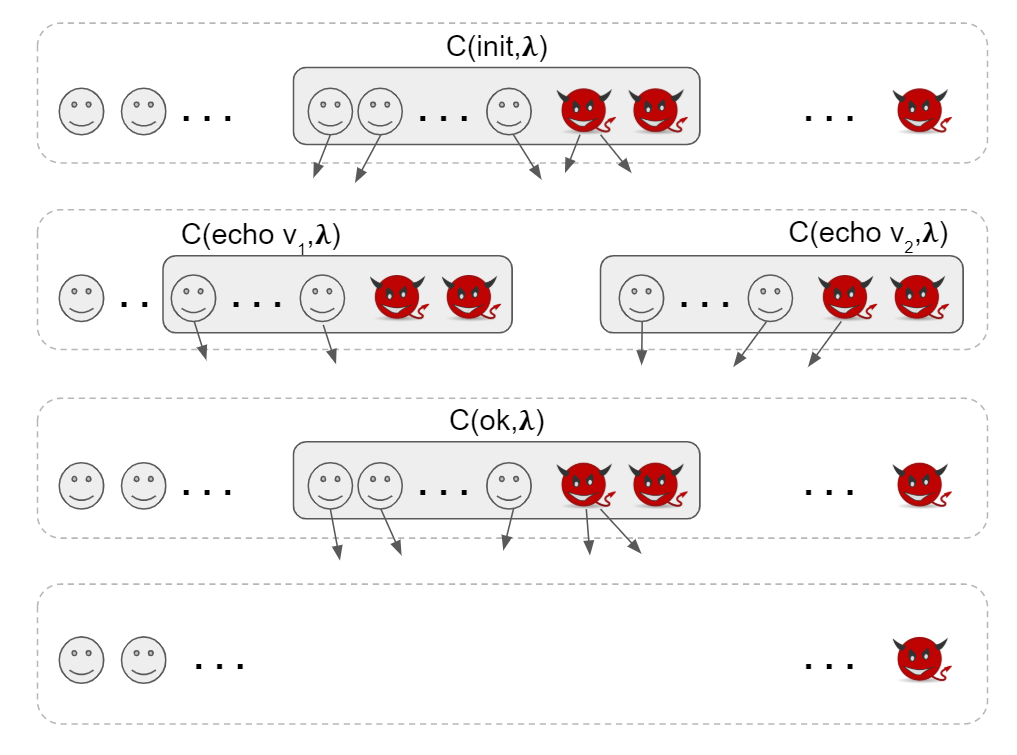}
\caption{Committees sampled in Algorithm \ref{alg:ratifier_comm}.}
\label{fig:comm_in_app}
\end{figure} 

The approver’s pseudo-code appears in Algorithm \ref{alg:ratifier_comm}.
It consists of three phases – init, echo, and ok. In each phase, committee members broadcast to all processes.
Messages are validated to originate from legitimate committee members using the  $\emph{committee-val}$ primitive; this validation is omitted from the pseudo-code for clarity.
In the first phase, each init committee member broadcasts its input value to all processes.

The role of the echo phase is to ``boost’’ values sent by sufficiently many processes in the init phase, and make sure that \emph{all} correct processes receive them. “Sufficiently many” here means at least  $B+1$, which by S4 includes at least one correct process.
Ensuring process replaceability in the echo phase is a bit tricky, since committee members must echo every value they receive from least  $B+1$ processes, and there might be two such values.
(Recall that we assume that correct processes invoke the protocol with at most two different values, so there cannot be more than two values that exceed this threshold).
To ensure that each committee member broadcasts at most once, we sample a different committee for each value. That is, the value $v$ is part of the string passed to the sample function for this phase.

When a member of the ok committee receives $\lr{\textsc{echo},v}$ messages from $W$ different members of the same echo committee for the first time, it broadcasts an $\lr{\textsc{ok},v}$ message. Note that the process sends an ok message only for the first value that exceeds this threshold. An $\lr{\textsc{ok},v}$ message includes, as proof of its validity, $W$ signed $\lr{\textsc{echo},v}$ messages. Again, the proof and its validation are omitted from the pseudo-code for clarity.
Once a correct process receives $W$ valid \textsc{ok} messages, it returns the set of values in these messages.

\begin{algorithm}
\caption{Protocol approve($v_i$): code for process $p_i$}


\begin{algorithmic}[1]

\If{$\emph{sample}_i(\textsc{init},\lambda) = $ \emph{true}}
	 broadcast $\lr{\textsc{init},v_i}$ \label{l.app.comm1}
\EndIf
\Statex

\Receiving{$\lr{\textsc{init},v}$ from $B+1$ different processes} 
    
\If{$\emph{sample}_i(\lr{\textsc{echo},v},\lambda) = $ \emph{true}}  broadcast $\lr{\textsc{echo},v}$ \label{l.app.comm2_and_echo}
\EndIf
\EndReceiving
\Statex

\Receiving{$\lr{\textsc{echo},v}$ from $W$ different processes}
\If{$\emph{sample}_i(\textsc{ok},\lambda) = $ \emph{true} $\land$ haven't sent any $\lr{\textsc{ok},*}$ message} \label{l.app.comm3}
    \State broadcast $\lr{\textsc{ok},v}$
\EndIf
\EndReceiving
\Statex

\Receiving{$\lr{\textsc{ok},*}$ from $W$ different processes}

    \State return the set of values received in these messages
\EndReceiving

\end{algorithmic}

\label{alg:ratifier_comm}
\end{algorithm}

We next prove that Algorithm \ref{alg:ratifier_comm} implements an approver. The following three lemmas are stated here and their proofs appear in Appendix \ref{approver_appendix}:

\begin{restatable}[Validity]{lemma}{approvervalidity}
\label{rat_validity_comm}
If all correct processes invoke \emph{approve(v)} then the only possible return value of correct processes is $\{v\}$ whp.
\end{restatable}

\begin{restatable}[Graded Agreement]{lemma}{approveragreement}
\label{rat_graded_agreement_comm}
If a correct process $p_i$ returns $\{v\}$ and another correct process $p_j$ returns $\{w\}$ then $v=w$ whp.
\end{restatable}

\begin{restatable}[Termination]{lemma}{approvertermination}

\label{rat_termination_comm}
If all correct processes invoke approve then at every correct process approve returns with a non-empty set whp.
\end{restatable}

From Lemmas \ref{rat_validity_comm},\ref{rat_graded_agreement_comm},\ref{rat_termination_comm}, we conclude the following theorem:

\begin{theorem}
Algorithm \ref{alg:ratifier_comm} implements an approver.
\end{theorem}

{\bf Complexity.} 
In each approver instance correct processes that are sampled to the four committees (lines \ref{l.app.comm1},\ref{l.app.comm2_and_echo},\ref{l.app.comm3}) send messages to all other processes.
The committee size is $O(\lambda)=O(\log n)$ whp.
Messages contain values, VRF proofs of the sender's election to the committee, signatures of $O(\lambda)$ committee members, and a constant number of bits that identify the type of message that is sent. Therefore, each message's size is at most $O(\lambda)$ words and the total word complexity of a shared coin instance is $O(n\lambda^2)=O(n\log^2 n)=\widetilde{O}(n)$ in expectation. The $\lambda^2$ appears in the expression due to the signatures of $O(\lambda)$ processes sent along the ok messages.

\subsection{Byzantine Agreement WHP}
\label{sec:consensus_protocol}
Our next step is solving Byzantine Agreement whp, formally defined as follows:

\begin{definition}[Byzantine Agreement WHP]
\label{def:consensus_whp}
In Byzantine Agreement WHP, each correct process $p_i\in\Pi$ proposes a binary input value $v_i$ and decide on an output value $decision_i$ s.t. with high probability the following properties hold:

\begin{itemize}
    \item Validity. If all correct processes propose the same value $v$, then any correct process
that decides, decides $v$.
    \item Agreement. No two correct processes decide differently.
    \item Termination. Every correct process eventually decides.
\end{itemize}

\end{definition}


We present the pseudo-code for our algorithm in Algorithm \ref{alg:consensus_protocol}.
Our protocol executes in rounds. Each round consists of two approver invocations and one call to the WHP coin.
Again, we discuss the algorithm assuming S1-S6 hold. We will argue that the algorithm decides in a constant number of rounds whp, and so these properties hold for all the committees it uses.
The local variable $est_i$ holds $p_i$'s current estimate of the decision value. The variable $decision_i$ holds $p_i$'s irrevocable decision. It is initialized to $\bot$ and set to a value in $\{0,1\}$ at most once.
Every process $p_i$ begins by setting $est_i$ to hold its initial value.
At the beginning of each round processes execute the approver with their $est$ values. If they return a singleton $\{v\}$, they choose to invoke the next approver with $v$ as their proposal and otherwise they invoke the next approver with $\bot$.
By the approver's graded agreement property, different processes do not return different singletons. Thus, at most two different values ($\bot$ and one in $\{0,1\}$) are given as an input by correct processes to the next approver, satisfying Assumption \ref{two_values}.

At this point, after all correct processes have chosen their proposals, they all invoke the WHP coin in line \ref{l.con.coin} in order to select a fall-back value. Notice that executing the WHP coin protocol after proposals have been set prevents the adversary from biasing proposals based on the coin flip. Then, in in line \ref{l.con.app2}, all processes invoke the approver with their proposals. If a process does not receive $\bot$ in its return set, it can safely decide the value it received. It does so by updating its $decision$ variable in line \ref{l.con.decide}. If it receives some value other than $\bot$ it adopts it to be its estimated value (line \ref{l.con.update_est3}), whereas if it receives only $\bot$, it adopts the coin flip (line \ref{l.con.update_est2}).
If all processes receive $\bot$ in line \ref{l.con.app1} then the probability that they all adopt the same value is at least $2\rho$, where $\rho$ is the coin's success rate. If some processes receive $v$, then the probability that all the processes that adopt the coin flip adopt $v$ is at least $\rho$.
With high probability, after a constant number of rounds, all correct processes have the same estimated value. By validity of the approver, once they all have common estimate, they decide upon it within 1 round.

\begin{algorithm}
\caption{Protocol Byzantine Agreement($v_i$): code for process $p_i$}

\begin{multicols}{2}

\begin{algorithmic}[1]

\State  $est_i \gets v_i$
\State $decision_i\gets \bot$
\Statex

\For{$r=0,1,...$}

\State  $vals \gets$ approve($est_i$) \label{l.con.app1}
\If{$vals=\{v\}$ for some $v$}
\State $propose_i\gets v$ \label{l.con.update_propose}
\EndIf
\State \textbf{otherwise} $propose_i\gets \bot$
\Statex

\State $c \gets $ whp\textunderscore coin($r$) \label{l.con.coin}
\State  $props \gets$ approve($propose_i$) \label{l.con.app2}
\If{$props=\{v\}$ for some $v\neq \bot$}
\State $est_i\gets v$ \label{l.con.update_est1}

\If{$decision_i=\bot$}
\State $decision_i\gets v$ \label{l.con.decide}
\EndIf

\Else \If{$props=\{ \bot\}$}
\State $est_i\gets c$ \label{l.con.update_est2}

\Else \Comment{$props=\{v,\bot \}$}
\State $est_i \gets v$ \label{l.con.update_est3}

\EndIf
\EndIf
\EndFor

\end{algorithmic}
\end{multicols}
\label{alg:consensus_protocol}
\end{algorithm}

We now prove our main theorem:

\begin{theorem}
\label{main_theo}
Algorithm~\ref{alg:consensus_protocol} when using an approver (Definition \ref{def:approver}) and a WHP coin (Definition \ref{def:whp_coin})
solves Byzantine Agreement whp (Definition \ref{def:consensus_whp}).
\end{theorem}

We first show that Algorithm \ref{alg:consensus_protocol} satisfies the approver and WHP coin primitives' assumptions whp. Proving this allows us to use their properties while proving Theorem \ref{main_theo}.

\begin{lemma}
\label{correct_use_of_primitives}
For every round $r$ of Algorithm \ref{alg:consensus_protocol} the following hold:
\begin{enumerate}
\item All correct processes invoke approve with at most 2 different values.
\item The invocation of whp\textunderscore coin$(r)$ by a correct process $p$ is causally independent of its progress at other processes.
\end{enumerate}

\end{lemma}
\begin{proof}
\begin{enumerate}
\item It is easy to see, by induction on the number of rounds, that since the processes' inputs are binary and we use a binary coin, the $est$ of all processes is in $\{0,1\}$ at the beginning of each round. Hence, the approver in line \ref{l.con.app1} is invoked with at most two different values. Due to its graded agreement property, all processes that update their propose to $v\neq \bot$ in line \ref{l.con.update_propose} update it to the same value whp. Thus, whp, in line \ref{l.con.app2} approver is invoked with either $v$ or $\bot$.

\item Correct processes call $whp$\textunderscore $coin(r)$ without waiting for indication that other processes have done so.
\end{enumerate}
\end{proof}

Next, we show that for any given round of the algorithm, (1) whp all processes complete this round, and (2) with a constant probability, they all have the same estimate value by its end.

\begin{lemma}
\label{complete_same_est}
If all correct processes begin round $r$ of Algorithm \ref{alg:consensus_protocol} then whp:
\begin{enumerate}
    \item All correct processes complete round $r$, i.e. they're not blocked during round $r$. \label{complete_same_est_1}
    \item With probability greater than $\rho$, where $\rho$ is the success rate of the WHP coin, all correct processes have the same $est$ value at the end of round $r$. \label{complete_same_est_2}
\end{enumerate}

\end{lemma}
\begin{proof}
First, if all correct processes begin round $r$ then they all invoke the approver in line \ref{l.con.app1}. Their invocation returns whp so they all invoke the coin in line \ref{l.con.coin}, and so it returns and all invoke approve in line \ref{l.con.app2}, and so it also returns, proving~(\ref{complete_same_est_1}). To show~(\ref{complete_same_est_2}), consider the possible scenarios with respect to the approver's return value:
\begin{itemize}
    \item All correct processes return singletons in line \ref{l.con.app1}:
    
    By the approver's graded agreement, whp they return $\{v\}$ with the same value $v$. Hence, all correct processes update their $propose$ to $v$. Then, they all execute approve($v$) in line \ref{l.con.app2}, and by validity, they all return $\{v\}$ whp. In this case they all update $est\gets v$.
    
    \item
    All correct processes return $\{0,1\}$ in line \ref{l.con.app1}:
    
    All correct processes update their $propose$ value to $\bot$. Then,
    they all execute approve($\bot$) in line \ref{l.con.app2}, and by validity, they  return $\{\bot\}$ whp. In this case, all correct processes then update their estimate value to the coin flip (line \ref{l.con.update_est2}). With probability at least $2\rho$ all correct processes toss the same $v\in \{0,1\}$.

    \item Some, but not all correct processes return singletons in line \ref{l.con.app1}:
    
    By graded agreement, all singletons hold the same value $v$. Thus, all correct processes propose $v$ or $\bot$ and by validity return $\{v\}, \{v,\bot\}$, or $\{\bot\}$ in line \ref{l.con.app2}.
    We examine two possible complementary sub-cases:
    
    \begin{itemize}
        \item If some correct process returns $\{v\}$ in line \ref{l.con.app2}:
        By approver's graded agreement, no correct process returns $\{\bot\}$ in line \ref{l.con.app2}, whp. Thus, whp, all correct processes update their estimate value to $v$ (in line \ref{l.con.update_est1} or \ref{l.con.update_est3}).
        
        \item If no correct process return $\{v\}$ in line \ref{l.con.app2}:
        All correct processes returns $\{v,\bot\}$ or $\{\bot\}$ in line \ref{l.con.app2}.
        All correct processes either update their estimate value to the coin flip of the WHP coin (line \ref{l.con.update_est2}) or to $v$ (line \ref{l.con.update_est3}). Since the value $v$ is determined before tossing the coin, the adversary cannot bias $v$ after viewing the coin flip and with probability at least $\rho$ all correct processes that adopt the coin's value toss $v$.
        
    \end{itemize}
\end{itemize}

In all cases, with probability greater than $\rho$ all correct processes have the same $est$ value at the end of $r$, whp.
\end{proof}

The following lemmas indicate that the Byzantine Agreement whp properties are satisfied, which completes the proof of Theorem \ref{main_theo}.

\begin{lemma}{(Validity)}
\label{validity_whp}
If at the beginning of round $r$ of Algorithm \ref{alg:consensus_protocol} all correct processes have the same estimate value $v$, then whp any correct process that has not decided before decides $v$ in round $r$.
\end{lemma}
\begin{proof}
If all correct processes start round $r$ then by Lemma \ref{complete_same_est} they all complete round $r$. Since they all being with the same estimate value $v$, they all execute approve($v$) in line \ref{l.con.app1}. Hence, by approver's validity and termination, whp they all return the non-empty set $\{v\}$ and update their $propose$ values to $v$. Then, they all execute approve($v$) for the second time in line \ref{l.con.app2}, and due to the same reason, they all return $\{v\}$ whp. Any correct process that has not decided before decides $v$ in line \ref{l.con.decide}.
\end{proof}

\begin{lemma}{(Termination)}
\label{termination_whp}
Every correct process decides whp.

\end{lemma}
\begin{proof}

By Lemma \ref{complete_same_est}, for every round $r$ of Algorithm \ref{alg:consensus_protocol}, with probability greater than $\rho$, where $\rho$ is the success rate of the WHP coin, all correct processes have the same $est$ value at the end of $r$ whp.
Hence, by Lemma \ref{validity_whp}, with probability greater than $\rho$, all correct processes decide by round $r+1$ whp. 
It follows that the expected number of rounds until all processes decide is bounded by $\frac{1}{\rho}$, which is constant.
Thus, by Chebyshev's inequality, whp all correct processes decide within a constant number of rounds.
\end{proof}

\begin{lemma}{(Agreement)}
\label{agreement_whp}
No two correct processes decide different values whp.
\end{lemma}
\begin{proof}
Let $r$ be the first round in which some process $p_i$ decides on some value $v\in \{0,1\}$. Thus, $p_i$'s invocation to approver in line \ref{l.con.app2} of round $r$ returns $\{v\}$. If another correct process $p_j$ decides $w$ in round $r$ then its approver call in line \ref{l.con.app2} of round $r$ returns $\{w\}$. By approver's graded agreement, $v=w$ whp.
Consider a correct process $p_k$ that does not decide in round $r$. By the definition of $r$, $p_k$ hasn't decided in any round $r'<r$. By approver's graded agreement, whp, $p_k$ returns $\{v,\bot\}$ in line \ref{l.con.app2} of round $r$, and $p_k$ updates its $est_{k}$ value to $v$ in line \ref{l.con.update_est3}.
It follows that whp all correct processes have $v$ as their estimate value at the beginning of round $r+1$. By Lemma \ref{validity_whp}, every correct process that has not decided in round $r$ decides $v$ in round $r+1$ whp.
\end{proof}

{\bf Complexity.} 
In each round of the protocol, all correct processes invoke two approver calls and one WHP coin instance. Due to the constant success rate of the WHP coin, the expected number of rounds before all correct processes decide is constant. Thus, due to the word complexity of the WHP coin and approver, the expected word complexity is $O(n\log^2 n)=\widetilde{O}(n)$ and the time complexity is $O(1)$ in expectation.

\section{Conclusions and Future Directions}
\label{sec:conclusions}

We have presented the first sub-quadratic asynchronous Byzantine Agreement algorithm.
To construct the algorithm, we introduced two techniques. First, we presented a shared coin algorithm that requires a trusted PKI and uses VRFs. Second, we formalized VRF-based committee sampling in the asynchronous model for the first time.

Our algorithm solves Byzantine Agreement with high probability. It would be interesting to understand whether some of the problem's properties can be satisfied with probability 1, while keeping the sub-quadratic communication cost. In addition, in order to achieve the constant success rate of the coin and guarantee the committees' properties, we bounded $\epsilon$ from below by a constant. This bound prevented us from achieving optimal resilience. The question whether it is possible to relax this bound to allow better resilience remains open.

\section*{Acknowledgements}
We thank Ittai Abraham, Dahlia Malkhi, Kartik Nayak and Ling Ren for insightful initial discussions.

\bibliography{references}

\appendix
 \newpage
\begin{appendices}
\section{Sampling proofs}
\label{chernoof_appendix}

\resampling*

\begin{proof}

Recall that $d$ is a parameter of the system such that $\max \{\frac{1}{\lambda},0.0362\}<d < \frac{\epsilon}{3}-\frac{1}{3\lambda}$.

In order to prove these properties we use two Chernoff bounds:

Suppose $X_1,...,X_n$ are independent random variables taking values in $\{0,1\}$. Let $X$ denote their sum and let $E[X]$ denote the sum's expected value.

\begin{equation}
\label{chernoff}
\forall 0 \leq \delta: \; Pr[X\geq (1+\delta)E[X]]\leq e^{-\frac{\delta ^2E[X]}{2+\delta}}
\end{equation}
\begin{equation}
\label{chernoff+}
\forall 0 \leq \delta\leq 1: \; Pr[X\leq (1-\delta)E[X]]\leq e^{-\frac{\delta ^2E[X]}{2}}
\end{equation}

\begin{lemma}[S1]
\label{comm_size_upper}
$|C(s,\lambda)|\leq (1+d)\lambda$ whp.
 \end{lemma}
\begin{proof}

Let $X$ be a random variable that represents the number of  processes that are sampled to $C(s,\lambda)$.
$X\sim Bin(n,\frac{const\cdot \ln n)}{n})$, thus $E[X]=const\cdot \ln n$.

Placing $\delta=d \geq 0$ in \ref{chernoff} we get:

\begin{equation*}
Pr[X\geq (1+d)const\cdot \ln n]\leq e^{-\frac{d^2const\cdot \ln n}{2+d}}.
\end{equation*}

Denote by $c_1$ the constant
$\frac{const\cdot d^2}{2+d}$. We get: 

\begin{equation*}
Pr[X\geq (1+d)const\cdot \ln n]\leq e^{-c_1\ln n}.
\end{equation*}

Thus, 
\begin{equation*}
Pr[X< (1+d)const\cdot \ln n]=Pr[X< (1+d)\lambda] > 1-\frac{1}{e^{c_1\ln n}}=1-\frac{1}{n^{c_1}}.
\end{equation*}

\end{proof}

\begin{lemma}[S2]
\label{comm_size_lower}
$|C(s,\lambda)|\geq (1-d)\lambda$ whp.
 \end{lemma}
\begin{proof}

Let $X$ be a random variable that represents the number of  processes that are sampled to $C(s,\lambda)$.
$X\sim Bin(n,\frac{const\cdot \ln n}{n})$, thus $E[X]=const\cdot \ln n$.

Placing $\delta=d $ it holds that $0\leq\delta\leq 1$ in \ref{chernoff+} and we get:

\begin{equation*}
Pr[X\geq (1-d)const\cdot \ln n]\leq e^{-\frac{d^2const\cdot \ln n}{2}}.
\end{equation*}

Denote by $c_2$ the constant
$\frac{const\cdot d^2}{2}$. We get: 

\begin{equation*}
Pr[X\geq (1-d)const\cdot \ln n]\leq e^{-c_2\ln n}.
\end{equation*}

Thus, 
\begin{equation*}
Pr[X<(1-d)const\cdot \ln n]=Pr[X< (1-d)\lambda] > 1-\frac{1}{e^{c_2\ln n}}=1-\frac{1}{n^{c_2}}.
\end{equation*}

\end{proof}

\begin{lemma}[S3]
At least $W$ processes in $C(s,\lambda)$ are correct whp.

\end{lemma}

\begin{proof}

Let $X$ be a random variable that represents the number of correct processes that are sampled to $C(s,\lambda)$.
$X\sim Bin((\frac{2}{3}+\epsilon)n,\frac{const\cdot \ln n}{n})$, thus $E[X]=(\frac{2}{3}+\epsilon)const\cdot \ln n$.
Let $d'=3d+\frac{1}{\lambda}$.
Notice that 
$1-\frac{\frac{2}{3}+d'}{\frac{2}{3}+\epsilon}\leq 1$ and also
$1-\frac{\frac{2}{3}+d'}{\frac{2}{3}+\epsilon}=1-\frac{\frac{2}{3}+3d+\frac{1}{\lambda}}{\frac{2}{3}+\epsilon}=\frac{\frac{2}{3}+\epsilon-\frac{2}{3}-3d-\frac{1}{\lambda}}{\frac{2}{3}+\epsilon} \geq
\frac{\frac{3}{\lambda}+\frac{1}{\lambda}-3d-\frac{1}{\lambda}}{\frac{2}{3}+\epsilon}
\geq 0$.
Hence, we can put $\delta=1-\frac{\frac{2}{3}+d'}{\frac{2}{3}+\epsilon}$ in (\ref{chernoff+}) and get:

\begin{equation*}
Pr[X\leq (1-(1-\frac{\frac{2}{3}+d'}{\frac{2}{3}+\epsilon}))(\frac{2}{3}+\epsilon)const\cdot \ln n]\leq e^{-\frac{(1-\frac{\frac{2}{3}+d'}{\frac{2}{3}+\epsilon}) ^2(\frac{2}{3}+\epsilon)const\cdot \ln n}{2}},
\end{equation*}

\begin{equation*}
Pr[X\leq (\frac{\frac{2}{3}+d'}{\frac{2}{3}+\epsilon})(\frac{2}{3}+\epsilon)const\cdot \ln n]\leq e^{-\frac{(1-\frac{\frac{2}{3}+d'}{\frac{2}{3}+\epsilon}) ^2(\frac{2}{3}+\epsilon)const\cdot \ln n}{2}},
\end{equation*}

\begin{equation*}
Pr[X\leq (\frac{2}{3}+d')const\cdot \ln n]\leq e^{-\frac{(1-\frac{\frac{2}{3}+d'}{\frac{2}{3}+\epsilon}) ^2(\frac{2}{3}+\epsilon)const\cdot \ln n}{2}}.
\end{equation*}

Denote by $c_3$ the constant $\frac{const\cdot (1-\frac{\frac{2}{3}+d'}{\frac{2}{3}+\epsilon}) ^2(\frac{2}{3}+\epsilon)}{2}$. We get: 

\begin{equation*}
Pr[X\leq (\frac{2}{3}+d')const\cdot \ln n]\leq e^{-c_3\ln n}.
\end{equation*}

Thus, 
\begin{equation*}
Pr[X> (\frac{2}{3}+d')const\cdot \ln n]=Pr[X> (\frac{2}{3}+d')\lambda] > 1-\frac{1}{e^{c_3\ln n}}=1-\frac{1}{n^{c_3}}.
\end{equation*}

To this point we've proved that at least
  $(\frac{2}{3}+d')\lambda$ processes in $C(s,\lambda)$ are correct whp. It follows that at least
  $(\frac{2}{3}+3d+\frac{1}{\lambda})\lambda=(\frac{2}{3}+3d)\lambda+1$ processes in $C(s,\lambda)$ are correct whp.

  As $\left \lceil{(\frac{2}{3}+3d)\lambda}\right \rceil \leq (\frac{2}{3}+3d)\lambda+1$ we conclude that at least
  $W=\left \lceil{(\frac{2}{3}+3d)\lambda}\right \rceil$ processes in $C(s,\lambda)$ are correct whp.

\end{proof}

\begin{lemma}[S4]
 At most $B$ processes in $C(s,\lambda)$ are Byzantine whp.
\end{lemma}
\begin{proof}

Let $X$ be a random variable that represents the number of Byzantine processes that are sampled to $C(s,\lambda)$.
$X\sim Bin((\frac{1}{3}-\epsilon)n,\frac{const\cdot \ln n}{n})$, thus $E[X]=(\frac{1}{3}-\epsilon)const\cdot \ln n$.

Placing $\delta=\frac{\epsilon-d}{\frac{1}{3}-\epsilon}\geq 0$ in (\ref{chernoff}) we get:

\begin{equation*}
Pr[X\geq (1+\frac{\epsilon-d}{\frac{1}{3}-\epsilon})(\frac{1}{3}-\epsilon)const\cdot \ln n]\leq e^{-\frac{(\frac{\epsilon-d}{\frac{1}{3}-\epsilon})^2(\frac{1}{3}-\epsilon)const\cdot \ln n}{2+(\frac{\epsilon-d}{\frac{1}{3}-\epsilon})}},
\end{equation*}
\begin{equation*}
Pr[X\geq (\frac{\frac{1}{3}-d}{\frac{1}{3}-\epsilon})(\frac{1}{3}-\epsilon)const\cdot \ln n]\leq e^{-\frac{\frac{(\epsilon-d)^2}{\frac{1}{3}-\epsilon}const\cdot \ln n}{2+(\frac{\epsilon-d}{\frac{1}{3}-\epsilon})}},
\end{equation*}
\begin{equation*}
Pr[X\geq (\frac{1}{3}-d)const\cdot \ln n]\leq e^{-\frac{\frac{(\epsilon-d)^2}{\frac{1}{3}-\epsilon}const\cdot \ln n}{2+(\frac{\epsilon-d}{\frac{1}{3}-\epsilon})}}.
\end{equation*}

Denote by $c_4$ the constant
${\frac{const\cdot \frac{(\epsilon-d)^2}{\frac{1}{3}-\epsilon}}{2+(\frac{\epsilon-d}{\frac{1}{3}-\epsilon})}}$. We get: 

\begin{equation*}
Pr[X\geq (\frac{1}{3}-d)const\cdot \ln n]\leq e^{-c_4\ln n}.
\end{equation*}

Thus, 
\begin{equation*}
Pr[X< (\frac{1}{3}-d)const\cdot \ln n]=Pr[X< (\frac{1}{3}-d)\lambda] > 1-\frac{1}{e^{c_4\ln n}}=1-\frac{1}{n^{c_4}}.
\end{equation*}

Since X must be an integer, it follows that $X\leq B=\left \lfloor{(\frac{1}{3}-d)\lambda}\right \rfloor$ whp.

\end{proof}

\end{proof}

\WW*

\begin{proof}
The set $P_2$ contains at most $|C(s,\lambda)\setminus P_1|$ processes that aren't in $P_1$. By S1, and since $P_1\subset C(s,\lambda)$:
\begin{equation*}
|C(s,\lambda)\setminus P_1|\leq
(1+d)\lambda-W =
(1+d)\lambda-\left \lceil{(\frac{2}{3}+3d)\lambda}\right \rceil \leq (1+d)\lambda-(\frac{2}{3}+3d)\lambda = (\frac{1}{3}-2d)\lambda.
\end{equation*}

The remaining processes in $P_2$ are also in $P_1$, so

\begin{equation*}
|P_1\cap P_2|\geq 
W-(\frac{1}{3}-2d)\lambda=
\left \lceil{(\frac{2}{3}+3d)\lambda}\right \rceil -(\frac{1}{3}-2d)\lambda
\geq 
(\frac{2}{3}+3d)\lambda -(\frac{1}{3}-2d)\lambda= (\frac{1}{3}+5d)\lambda.
\end{equation*}
Finally, 

\begin{equation*}
|P_1\cap P_2|-B =
|P_1\cap P_2|-\left \lfloor{(\frac{1}{3}-d)\lambda}\right \rfloor\geq (\frac{1}{3}+5d)\lambda-(\frac{1}{3}-d)\lambda=6d\lambda>\frac{6\lambda}{\lambda} \geq 1,
\end{equation*}

as requested.
\end{proof}

\WB*

\begin{proof}
    
The set $P_2$ contains at most $|C(s,\lambda)\setminus P_1|$ processes that aren't in $P_1$. By S1, and since $P_1\subset C(s,\lambda)$:
\begin{equation*}
|C(s,\lambda)\setminus P_1|\leq
(1+d)\lambda-(B+1) =
(1+d)\lambda-(\left \lfloor{(\frac{1}{3}-d)\lambda}\right \rfloor+1)
\leq
(1+d)\lambda-((\frac{1}{3}-d)\lambda-1)-1
=(\frac{2}{3}+2d)\lambda.
\end{equation*}

Therefore,

\begin{equation*}
|P_2|-|C(s,\lambda)\setminus P_1|\geq 
W-(\frac{2}{3}+2d)\lambda=
\left \lceil{(\frac{2}{3}+3d)\lambda}\right \rceil-(\frac{2}{3}+2d)\lambda
\geq
{(\frac{2}{3}+3d)\lambda}-(\frac{2}{3}+2d)\lambda
=d\lambda >\frac{\lambda}{\lambda}=1 ,
\end{equation*}

and so $|P_1\cap P_2|\geq 1$, as requested.
\end{proof}

\section{WHP Coin Proofs}
\label{whp_coin_appendix}

In the committee-based protocol, a value $v$ is \emph{common} if at least $B+1$ correct processes in $C(\textsc{second},\lambda)$ have $v_i=v$ at the end of phase 1. The next lemma adapts the lower bound of Lemma \ref{common_vals} on the number of common values to the committee-based protocol.
 
\begin{lemma}
\label{common_vals_comm}
In Algorithm \ref{alg:shared_coin_comm_protocol} whp, $c \geq \frac{d(11-3d)}{1+9d}\lambda$.
\end{lemma}

\begin{proof}

Let $n_1=|C(\textsc{first},\lambda)|,n_2=|C(\textsc{second},\lambda)|$.
We define a table T with $n_2$ rows and $n_1$ columns. 
For each correct process $p_i\in C(\textsc{second},\lambda)$ and each $0\leq j\leq n_1-1$, $T[i,j]=1$ iff $p_i$ receives $\lr{\textsc{first},v}$ from $p_j\in P_1$ before sending the \textsc{second} message in line \ref{l.com.send_second}.
Each row of a correct process contains exactly
$W$ ones since it waits for $W$ $\lr{\textsc{first},v}$ messages (line \ref{l.com.wait_first}). Each row of a faulty process in $C(\textsc{second},\lambda)$ is arbitrarily filled with $W$ ones and $n_1-W$ zeros.
Thus the total number of ones in the table is $n_2W$ and the total number of zeros is $n_2(n_1-W)$.
Let $k$ be the number of columns with at least $2B+1$ ones.
Each column represents a value sent by a process in $C(\textsc{first},\lambda)$.
By S4, whp, at most $B$ of the processes that receive this value are Byzantine.
Thus, whp, out of any $2B+1$ ones in each of these columns, at least $B+1$ represent correct processes that receive this value and it follows that $c\geq k$.

Denote by $x$ the number of ones in the remaining columns. Because each column has at most $n_2$ ones we get:

\begin{equation}
x \geq n_2W-kn_2 =
n_2 \left \lceil{(\frac{2}{3}+3d)\lambda}\right \rceil -kn_2
\geq n_2(\frac{2}{3}+3d)\lambda-kn_2.
\end{equation}

And because the remaining columns have at most $2B$ ones:
\begin{equation}
x \leq 2B(n_1-k)
= 2\left \lfloor{(\frac{1}{3}-d)\lambda}\right \rfloor(n_1-k)
\leq 2(\frac{1}{3}-d)\lambda(n_1-k).
\end{equation}

Combining $(1),(2)$ we get:

\begin{equation*}
2(\frac{1}{3}-d)\lambda(n_1-k) \geq
n_2(\frac{2}{3}+3d)\lambda-kn_2
\end{equation*}

\begin{equation*}
kn_2-2\lambda k(\frac{1}{3}-d) \geq
n_2(\frac{2}{3}+3d)\lambda-2(\frac{1}{3}-d)\lambda n_1
\end{equation*}

\begin{equation*}
k(n_2-2\lambda (\frac{1}{3}-d)) \geq \lambda(
n_2(\frac{2}{3}+3d)-2(\frac{1}{3}-d)n_1)
\end{equation*}

\begin{equation*}
k \geq \frac{\lambda(
n_2(\frac{2}{3}+3d)-2(\frac{1}{3}-d)n_1)}{n_2-2\lambda (\frac{1}{3}-d)}
\end{equation*}

By S2 for $C(\textsc{second},\lambda)$, whp $n_2\geq(1-d)\lambda$ and we get:

\begin{equation*}
k \geq \frac{\lambda(
(1-d)\lambda(\frac{2}{3}+3d)-2(\frac{1}{3}-d)n_1)}{n_2-2\lambda (\frac{1}{3}-d)}
\end{equation*}

By S1 for $C(\textsc{first},\lambda)$ and $C(\textsc{second},\lambda)$, whp $n_1,n_2\leq (1+d)\lambda$ and we get:

\begin{equation*}
k \geq \frac{\lambda \bigg[
(1-d)\lambda(\frac{2}{3}+3d)-2(\frac{1}{3}-d)(1+d)\lambda \bigg]}{(1+d)\lambda-2\lambda (\frac{1}{3}-d)}=\frac{\lambda\bigg[
(1-d)(\frac{2}{3}+3d)-2(\frac{1}{3}-d)(1+d)\bigg]}{(1+d)-2 (\frac{1}{3}-d)}
\end{equation*}

Finally, we get whp:

\begin{equation*}
c \geq k \geq \frac{d(11-3d)}{1+9d}\lambda.
\end{equation*}
as required.

\end{proof}

Let $v_{min}\triangleq \displaystyle \min_{ p_i \in C(\textsc{first},\lambda)}\{VRF_i(r)\}$. Similiarly to Lemma \ref{common_prob}, we prove that the probability that it is common is bounded by a constant, whp. I.e., we show that $Prob[v_{min}\;is\;common]\geq const\cdot g(n)$ where $g(n)$ goes to $1$ as $n$ goes to infinity.

\begin{lemma}
\label{common_prob_comm}
whp $Prob[v_{min}\;is\;common]\geq \frac{2}{3(1-d)}\cdot \frac{c-B}{(1+d)\lambda-B}$.
\end{lemma}

\begin{proof}

Notice that we assume that the invocation of whp\textunderscore coin$(r)$ by every process is causally independent of its progress at other processes. Hence, for any two processes $p_i, p_j\in C(\textsc{first},\lambda)$, the messages $\lr{\textsc{first},v_i}$, $\lr{\textsc{first},v_j}$ are causally concurrent.
Thus, due to our \emph{delayed-adaptive adversary} definition,
these messages are scheduled by the adversary regardless of their content, namely their VRF random values.
Notice that the adversary can corrupt processes before they initially send their VRF values.
By S4 there are at most $B$ Byzantine processes in $C(\textsc{first},\lambda)$. 
Since the adversary cannot predict the VRF outputs, the probability for a given process to be corrupted before sending its \textsc{first} messages is at most $\frac{B}{|C(\textsc{first},\lambda)|}$.
The adversary is oblivious to the correct processes' VRF values when it schedules their first phase messages. Therefore, each of them has the same probability to become common. Since at most $B$ common values are from Byzantine processes, this probability is at least $\frac{c-B}{|C(\textsc{first},\lambda)|-B}$.
 We conclude that $v_{min}$ is common with probability at least $(1-\frac{B}{|C(\textsc{first},\lambda)|})\frac{c-B}{|C(\textsc{first},\lambda)|-B}$.
By S1 and S2 we get that $(1-d)\lambda \leq |C(\textsc{first},\lambda)| \leq (1+d)\lambda$ whp.

Thus, whp, $v_{min}$ is common with probability at least $(1-\frac{B}{(1-d)\lambda})\frac{c-B}{(1+d)\lambda-B}=
(1-\frac{\left \lfloor{(\frac{1}{3}-d)\lambda}\right \rfloor}{(1-d)\lambda})\frac{c-B}{(1+d)\lambda-B} \geq
(1-\frac{(\frac{1}{3}-d)\lambda}{(1-d)\lambda})\frac{c-B}{(1+d)\lambda-B}=
\frac{2}{3(1-d)}\cdot \frac{c-B}{(1+d)\lambda-B}
$.

\end{proof}

\begin{lemma}
\label{common_global_min_comm}
If $v_{min}$ is common then whp each correct process holds $v_{min}$ at the end of phase 2.
\end{lemma}
\begin{proof}
Since $v_{min}$ is common, at least $B+1$ correct members of $C(\textsc{second},\lambda)$ receive it by the end of phase 1 and update their local values to $v_{min}$. During the second phase, each correct process hears from $W$ members of $C(\textsc{second},\lambda)$ whp. By S6, this means that it hears from at least one correct process that has updated its value to $v_{min}$ and sent it whp.
\end{proof}

\begin{lemma}
\label{safety_comm}
Let $\rho=\frac{18d^2+27d-1}{3(5+6d)(1-d)(1+9d)}$. Algorithm \ref{alg:shared_coin_comm_protocol} implements a shared coin with success rate $\rho$, whp.
\end{lemma}

\begin{proof}
Denote $n_1=|C(\textsc{first},\lambda)|$. We bound whp the probability that all correct processes output $b\in \{0,1\}$ as follows:

 $Prob[$all correct processes output $b]\geq
 Prob[$all\;correct\;processes\;have\;the\;same\;$v_i$ at the end of phase 2\;and\;its\;LSB\;is\;$b]\geq Prob[$all\;correct\;processes have $v_i=v_{min}$  at the end of phase 2 and its LSB is $b]=\frac{1}{2} \cdot Prob[$all\;correct\;processes have $v_i=v_{min}]\stackrel{\text{Lemma \ref{common_global_min_comm}}}{\geq}
 \frac{1}{2} \cdot Prob[v_{min}$ is common$]\stackrel{\text{Lemma \ref{common_prob_comm}}}{\geq} \frac{1}{2} \cdot \frac{2}{3(1-d)}\cdot \frac{c-B}{(1+d)\lambda-B} \stackrel{\text{Lemma \ref{common_vals_comm}}}{\geq}
 \frac{1}{3(1-d)}\cdot \frac{\frac{d(11-3d)}{1+9d}\lambda-B}{(1+d)\lambda-B} =
 \frac{1}{3(1-d)}\cdot \frac{\frac{d(11-3d)}{1+9d}\lambda-\left \lfloor{(\frac{1}{3}-d)\lambda}\right \rfloor}{(1+d)\lambda-\left \lfloor{(\frac{1}{3}-d)\lambda}\right \rfloor} \geq
 \frac{1}{3(1-d)}\cdot \frac{\frac{d(11-3d)}{1+9d}\lambda-(\frac{1}{3}-d)\lambda}{(1+d)\lambda-((\frac{1}{3}-d)\lambda-1)}=
 \frac{1}{3(1-d)}\cdot \frac{\lambda\frac{18d^2+27d-1}{27d+3}}{\lambda(\frac{2}{3}+2d)+1}\geq
 \frac{1}{3(1-d)}\cdot \frac{\lambda\frac{18d^2+27d-1}{27d+3}}{\lambda(\frac{2}{3}+2d)+\lambda}=
\frac{18d^2+27d-1}{3(5+6d)(1-d)(1+9d)}$.
 
\end{proof}

We have shown a bound on the coin's success rate whp. Since $d>0.0362$, the coin's success rate is a positive constant whp.
We next prove that the coin ensures liveness whp.
 
\begin{lemma}
\label{shared_coin_termination_comm}
If all correct processes invoke Algorithm \ref{alg:shared_coin_comm_protocol} then all correct processes return whp.
\end{lemma}

\begin{proof}
All correct processes in $C(\textsc{first},\lambda)$ send their message in the first phase.
At least $W$ of them are correct whp by S3. All correct processes in $C(\textsc{second},\lambda)$ eventually receive $W$ $\lr{\textsc{first},x}$ messages whp and send a message in the second phase. As ,whp, again $W$ correct processes send their messages (by S3), each correct process eventually receives $W$ $\lr{\textsc{second},x}$ messages and returns whp.
\end{proof}

From Lemma \ref{safety_comm} and Lemma \ref{shared_coin_termination_comm} we conclude:

\rewhpcoin*

\section{Approver proofs}
\label{approver_appendix}

\approvervalidity*
\begin{proof}
By Claim \ref{sampling} S4 holds whp for the four sampled committees. It remains to show that S4 implies validity. Since by S4 the number of Byzantine processes sampled to the init committee in line \ref{l.app.comm1} is at most $B$, no process receives $B+1$ messages with a value $w\neq v$. Thus, no correct process echoes $\lr{\textsc{echo},w}$ in line \ref{l.app.comm2_and_echo}. 
Because the number of Byzantine processes in $C(\lr{\textsc{echo},w},\lambda)$ in line \ref{l.app.comm2_and_echo} is also at most $B$, no correct process receives more than $B$ $\lr{\textsc{echo},w}$ messages.
As a result, since $B<W$, no $\lr{\textsc{ok},w}$ message is sent by any correct process.
Since ok messages carry proofs, no Byzantine process can send a valid $\lr{\textsc{ok},w}$ either. Therefore, the only possible value in the ok messages is $v$, and no other value is returned.
\end{proof}

\approveragreement*

\begin{proof}

By Claim \ref{sampling} and Corollary \ref{S5}, S4 and S5 hold whp for the four sampled committees. We show that S4 and S5 imply graded agreement.
Assume $p_i$ returns $\{v\}$ and $p_j$ returns $\{w\}$. Then $p_i$ receives $W$ $\lr{\textsc{ok},v}$ messages and $p_j$ receives $W$ $\lr{\textsc{ok},w}$ messages. 
By S5, two sets of size $W$ intersect by at least $(\frac{1}{3}-d)\lambda+1$ processes. Hence, since by S4 there are at most $B$ Byzantine processes in the ok committee, there is at least one correct process $p_k$ whose ok message is received by both $p_i$ and $p_j$ whp. It follows that $p_k$ sends $\lr{\textsc{ok},v}$ and $\lr{\textsc{ok},w}$. Since every correct process sends at most one ok message (line \ref{l.app.comm3}), $v=w$.
\end{proof}

\approvertermination*
\begin{proof}
By Claim \ref{sampling} S3 holds whp. We show that S3 implies termination. Because all correct processes invoke approve, every correct init committee member in line \ref{l.app.comm1} sends $\lr{\textsc{init},v_i}$.
Notice that $\frac{1}{2}W > (\frac{1}{3}-d)\lambda \geq B$. Hence, since the number of correct processes in the init committee is at least $W$ (S3) and correct processes may send at most two different initial values (Assumption \ref{two_values}), one of them is sent by at least $B+1$ correct processes.
Denote this value by $v$.
Every correct process receives this value from $B+1$ processes, and if it is sampled to $C(\lr{\textsc{echo},v},\lambda)$ in line \ref{l.app.comm2_and_echo} then it sends it to all other processes. 
Since $C(\lr{\textsc{echo},v},\lambda)$ also has at least $W$ correct processes (S3), every correct process $p$ receives $W$ $\lr{\textsc{echo},v}$ messages.
If $p$ is sampled to the ok committee in line \ref{l.app.comm3} and at this point $p$ has not yet sent an $\lr{\textsc{ok},*}$ message, it sends one.
Since there are at least $W$ correct processes that are sampled to the ok committee (S3) and they all send $\textsc{ok}$ messages (possibly for different values), every correct process receives $W$ $\textsc{ok}$ messages and returns the non-empty set of approved values.

\end{proof}
\end{appendices}

\end{document}